\documentclass{sig-alternate}

\usepackage{hyperref}
\usepackage{pgfplots}
\pgfplotsset{compat=1.3}

\usepackage[ruled, vlined, linesnumbered]{algorithm2e}
\usepackage{enumitem}
\usepackage{cite}

\usepackage{balance}
\usepackage{graphicx}
\usepackage{caption}
\usepackage{subcaption}

\newcommand{\Lap}{\operatorname{Lap}}
\newcommand{\E}{\mathbb{E}}
\newcommand{\Var}{\operatorname{Var}}
\newcommand{\p}{\prime}

\def\header{\vspace{2.5mm} \noindent}

\newtheorem{definition}{Definition}
\newtheorem{lemma}{Lemma}

        \newenvironment{customlegend}[1][]{        \begingroup
                        \csname pgfplots@init@cleared@structures\endcsname
        \pgfplotsset{#1}    }{                \csname pgfplots@createlegend\endcsname
        \endgroup
    }
            \def\addlegendimage{\csname pgfplots@addlegendimage\endcsname}
\pgfplotsset{
cycle list={{draw=black,mark=star,solid},
{draw=black, mark=square,solid},{draw=black,mark=+,solid},{black,mark=o},}}

\begin{document}

\title{Collecting and Analyzing Data from Smart Device Users with Local Differential Privacy}

\author{
    Th\^ong T. Nguy\~{\^e}n$^{1}$ \hspace{2mm} Xiaokui Xiao$^1$ \hspace{2mm} Yin Yang$^2$ \hspace{2mm} Siu Cheung Hui$^1$ \hspace{2mm} Hyejin Shin$^3$ \hspace{2mm} Junbum Shin$^3$ \\
    \and
    \alignauthor
    \affaddr{\hspace {3mm}$^1$Nanyang Technological University $\qquad \qquad \qquad \quad$ $^2$Hamad Bin Khalifa University} \\
    \email{\hspace {-12mm} \{s140046, xkxiao, asschui\}@ntu.edu.sg \hspace{20mm} yyang@qf.org.qa} \\[2mm]
    \affaddr{$^3$Samsung Electronics}
    \email{\{hyejin1.shin, junbum.shin\}@samsung.com}\\
}
\copyrightetc{ \the\acmcopyr}

\crdata{}

\maketitle

\begin{sloppy}

\begin{abstract}

Organizations with a large user base, such as Samsung and Google, can potentially benefit from collecting and mining users' data. However, doing so raises privacy concerns, and risks accidental privacy breaches with serious consequences. Local differential privacy (LDP) techniques address this problem by only collecting randomized answers from each user, with guarantees of plausible deniability; meanwhile, the aggregator can still build accurate models and predictors by analyzing large amounts of such randomized data. So far, existing LDP solutions either have severely restricted functionality, or focus mainly on theoretical aspects such as asymptotical bounds rather than practical usability and performance. Motivated by this, we propose Harmony, a practical, accurate and efficient system for collecting and analyzing data from smart device users, while satisfying LDP. Harmony applies to multi-dimensional data containing both numerical and categorical attributes, and supports both basic statistics (e.g., mean and frequency estimates), and complex machine learning tasks (e.g., linear regression, logistic regression and SVM classification). Experiments using real data confirm Harmony's effectiveness.
\end{abstract}

\section{Introduction} \label{sec:intro}

Smart devices connected to the Internet, including mobile phones, wearables, home appliances, sensors and vehicles, have become a part of everyday life in many parts of the world. The data collected by these devices could be an invaluable asset to hardware designers and application developers. For instance, a smartphone maker with its own customized UI such as Samsung TouchWiz could learn about the usage patterns of the various UI features such as Multi Window and One-Handed Mode, and focus on improving the popular ones. However, privacy concerns remain a major hurdle in collecting users' data. For example, the user may not want others to know her web browsing history, apps installed and locations visited. Even if the user (reluctantly) allows trusted organizations to collect her data, the possession of large amounts of sensitive personal data poses a major security risk. Accidental leakage of such personal data, which happened to AOL\footnote{\small \url{https://en.wikipedia.org/wiki/AOL_search_data_leak}}, Netflix\footnote{\small \url{http://www.wired.com/2009/12/netflix-privacy-lawsuit/}} and Ashley Madison\footnote{\small \url{http://edition.cnn.com/2015/08/27/opinions/yang-ashley-madison-hack/}}, led to serious consequences and substantial damage. So, organizations with a large user base commonly face a dilemma: either they collect users' personal data and become exposed to the risk of privacy breaches, or they do not collect such data and lose the opportunity of mining them.

Local differential privacy (LDP), which has been used in Google's Chrome browser \cite{RAPPOR2014}, addresses the above dilemma. The idea is to compute aggregates of users' data without collecting individuals precise personal information. Unlike other models of differential privacy \cite{DworkLaplace, dwork2004privacy}, which publish randomized aggregates but still collect the exact sensitive data, LDP avoids collecting exact personal information in the first place, thus providing a stronger assurance to the users and to the aggregator. Meanwhile, LDP satisfies the strong and rigorous privacy guarantees of differential privacy, i.e., the adversary cannot infer sensitive information of an individual with high confidence, regardless of the adversary's background knowledge.

Google's LDP solution in its Chrome browser, namely Rappor \cite{RAPPOR2014}, has rather limited functionalities. The core of Rappor is a randomized response mechanism \cite{Warner65} for a user to answer a yes/no question to the aggregator. A classic example is to collect statistics about a sensitive group (e.g., communists in the US), in which the aggregator asks each individual: ``Are you a communist?'' To answer this question, each individual tosses a coin, gives the true answer if it is a head, and a random yes/or answer otherwise. Clearly, this randomized approach provides plausible deniability to the individuals. Meanwhile, it is shown to satisfy $\epsilon$-differential privacy, and the strength of privacy protection (i.e., $\epsilon$) can be controlled by using an unfair coin \cite{RAPPOR2014}. Based on the collected randomized answers, the aggregator estimates the percentage of users whose true answer is ``yes'' (resp. ``no''). Besides simple counting, a follow-up paper \cite{FantiPE15} shows that Rappor can also compute other types of statistics such as joint-distribution estimation and association testing. However, three major limitations remain: first, Rappor cannot compute aggregates on numeric attributes, e.g., the average running time of an app. Second, the accuracy of Rappor deteriorates quickly with increasing number of attributes; as we show later in Section \ref{sec:basic-categorical}, to compute aggregates on $d$ independent attributes, the error of Rappor grows linearly with $d$, which is sub-optimal \cite{duchi2013local, DuchiJW14}. Third, it is unclear whether/how Rappor can handle complex and yet commonly used machine learning tasks, such as logistic regression and SVM classification.

LDP has also drawn considerable interest from the theory community. However, as we review in Section~\ref{sec:related}, their focus lies mostly in analyzing the asymptotical performance of basic building blocks of LDP, rather than practical systems and performance. In other words, the problem they study is ``what is the best that LDP can do'', rather than ``how to design a practical LDP solution''. Perhaps due to this reason, as we show in Section \ref{sec:basic-numeric}, the state-of-the-art solution (to our knowledge) for estimating mean values of multiple attributes under LDP is buggy and requires a fix. Meanwhile, we are not aware of any systematic approach that can perform common machine learning tasks under LDP on a large data domain containing both numeric and categorical attributes.

Motivated by this, we have been building Harmony, an advanced data analytics tool that conforms to LDP requirements. Harmony supports a multitude of common data analysis tasks over an arbitrary number of numerical or categorical attributes. Further, Harmony achieves both non-trivial asymptotical error bounds and improved accuracy in practice, compared to the current state of the art. As a case study, we show how Harmony can improve diagnostic information reporting in Samsung smartphones with strong privacy guarantees, as follows.

\subsection{Potential Use Case} \label{sec:usecase}

Samsung currently collects mobile phone usage information from users through a diagnostic tool bundled with the Samsung Android OS\footnote{The tool is available in the system under ``Settings > About device > Report diagnostic info''.}. The information collected include the mobile phone's settings (e.g., display settings, whether or not the location functionality is turned on), memory and battery usage, as well as other log data. The tool transmits the data in an unperturbed format to Samsung, and the transmission explicitly requires users' consent. However, as the collected information could potentially reveal sensitive information, it is beneficial to enhance the tool with privacy protection mechanisms, so as to provide rigorous privacy assurance to users.

Towards this end, local differential privacy is an attractive approach due to its strong privacy guarantee, and its practicability that has been demonstrated in Rappor \cite{RAPPOR2014}. However, Rappor's LDP mechanism focuses on collecting a single categorical attribute, whereas the data collected by the Samsung diagnostic tool include both numeric ones (e.g., battery usage) and categorical ones (e.g., display settings). Furthermore, Rappor does not support complex learning tasks (e.g., regressions, SVM), whereas such tasks are important for Samsung to build analytical models from the data, e.g., building an early symptom model to predict system errors based on other measurements such as battery usage, memory usage, active applications, etc. This necessitates the development of new data collection technique based on local differential privacy.

In the following, Section \ref{sec:prelim} provides the necessary background on LDP. Sections \ref{sec:basic} presents the fundamental LDP mechanisms in Harmony. Section \ref{sec:riskmini} applies Harmony to common data analytics tasks based on empirical risk minimization, including linear regression, logistic regression and SVM classification. Section \ref{sec:exp} contains an extensive set of experiments. Section \ref{sec:related} reviews related work. Finally, Section \ref{sec:conclusion} concludes with directions for future work.
 \section{Preliminaries} \label{sec:prelim}

\begin{table}[t]
\caption{Notations}
\centering
\label{tbl:notations}
\begin{tabular}{|l|l|}
\hline
$pdf(x)$        & Probability density at $x$          \\ \hline
$\langle a, b \rangle$            & A tuple of $a$ and $b$         \\ \hline
$a || b$              & A concatenation of $a$ and $b$        \\ \hline
$\operatorname{Bern}(p)$              & A Bernoulli distribution with parameter $p$        \\ \hline
$\log(x)$              & Natural logarithm of $x$        \\ \hline
$\operatorname{Dom}(f)$              & Domain of function $f$        \\ \hline
$[d]$              & Set of natural numbers $\{1, 2, \ldots, d\}$        \\ \hline
\end{tabular}
\end{table}

In our setting, an \emph{aggregator} (e.g., Samsung) collects data from a set of \emph{users} (e.g., smart device owners), and computes statistical models of the collected data. The goal is to maximize the accuracy of these statistical models, while preserving the privacy of the users. Following the local differential privacy model \cite{RAPPOR2014, BS15, DuchiJW14}, we assume that the aggregator already knows the identities (e.g., IP addresses) of the users, but not their private data. Formally, let $n$ be the total number of users, and $u_i$ ($1 \leq i \leq n$) denote the $i$-th user. Each user $u_i$'s private data is represented by a tuple $t_i$, which contains $d$ attributes $A_1, A_2, \ldots, A_d$. These attributes can be either numerical or categorical. Without loss of generality, we assume that each numeric attribute has a domain $[-1, 1]$, and each categorical attribute with $k$ distinct values has a discrete domain $\{1, 2, \ldots, k\}$.

To protect privacy, each user $u_i$ first perturbs her tuple $t_i$ using a randomized \emph{perturbation function} $f$. Then, she sends the perturbed data $f(t_i)$ to the aggregator instead of her true data record $t_i$. The perturbation function determines the privacy / utility tradeoff. As an extreme case, if $f(t_i) = t_i$ (i.e., no perturbation), the aggregator obtains perfect utility since it computes the statistical models based on the exact data; however, the privacy of the users is completely lost as the aggregator receives their sensitive data. The other extreme is that $f$ simply outputs a random tuple regardless of $t_i$, which leads to the highest level of privacy and zero utility, i.e., the aggregator learns nothing about the users' data.

Given a privacy parameter $\epsilon > 0$ that controls the privacy-utility tradeoff, we require that $f$ satisfies $\epsilon$-\emph{local differential privacy} ($\epsilon$-LDP) \cite{RAPPOR2014}, defined as follows:
\begin{definition}[$\epsilon$-local differential privacy]\label{def:ldp}
A randomized function $f$ satisfies $\epsilon$-local differential privacy if and only if for any two input tuples $t, t^\prime \in \text{Dom(f)}$ and for any possible output $t^*$ of $f$, we have:
 \begin{displaymath}
\Pr[f(t) = t^*] \leq e^\epsilon \times \Pr[f(t^\prime) = t^*].
 \end{displaymath}
\end{definition}

Basically, local differential privacy is a special case of differential privacy \cite{dwork2004privacy} where the random perturbation is performed by the users, not by the aggregator. In other words, the aggregator never possesses the exact private data of any user. According to the above definition, the aggregator, who receives the perturbed tuple $t^*$, cannot distinguish whether the true tuple is $t$ or another tuple $t'$ with high confidence (controlled by parameter $\epsilon$), regardless of the background information of the aggregator. This provides plausible deniability to the user. Note that in $\epsilon$-LDP, since random perturbation is done at each user, it is possible to achieve \emph{personalized privacy protection} by using different values of the privacy parameter $\epsilon$ at different users, depending on their respective privacy requirements. In this paper, we assume a universal $\epsilon$ for the ease of presentation and analysis.

We aim to support the following types of analytics tasks under $\epsilon$-LDP:
\begin{enumerate}[topsep = 6pt, parsep = 6pt, itemsep = 0pt, leftmargin=22pt]
\item {Mean value and frequency estimation.} These are two basic types of statistics. For each numeric attribute $A_j$, we aim to estimate the mean value of $A_j$ over all $n$ users, $\frac{1}{n} \sum_{i=1}^n {t_i[A_j]}$. For each categorical attribute $A_j'$, we aim to estimate the frequency of each possible value of $A_j'$, i.e., a histogram. Note that it is also possible to build such a histogram for a numeric attribute with a finite number of possible values.

\item {Empirical Risk Minimization.} These are advanced statistics commonly used in machine learning. Examples include linear regression, logistic regression, and support vector machines (SVM) \cite{CortesV95}.
\end{enumerate}
Unless otherwise specified, all expectations in this paper are taken over the random choices made by the algorithms considered.

\header
{\bf Remark.} In practice, the tuple $t_i$ of a user may change overtime (e.g., the phone usage information of a user would change day by day); accordingly, the aggregator may want to re-collect information from users after a certain time period (e.g., one week). In that case, we aim to ensure that the collection of each individual snapshot of $t_i$ satisfies $\epsilon$-differential privacy. One may argue that it is more desirable to ensure that all collected snapshots jointly achieve $\epsilon$-differential privacy, but to our knowledge, this is an open problem when the number of snapshots to be collected can be arbitrarily large.

 \section{Estimating Means and Frequencies} \label{sec:basic}

This section investigates the design of the perturbation function $f$ to support accurate estimation of mean values (resp.\ frequencies) of numeric (resp.\ categorical) attributes. For ease of exposition, Section~\ref{sec:basic-numeric} considers the case when all attributes $A_1, A_2, \ldots, A_d$ in the users' data have a numeric domain $[-1, 1]$; after that, Section~\ref{sec:basic-categorical} extends our discussions to the case when both numeric and categorical attributes are present.

\subsection{Estimating Mean Values for Numeric Attributes} \label{sec:basic-numeric}

Given a tuple $t_i$ containing $d$ numeric attributes, a naive design of the perturbation function $f$ is to apply the \emph{Laplace Mechanism} \cite{DworkLaplace}. In particular, let $t^*_i = f(t_i)$ be the perturbed tuple of user $u_i$, we have:
$$\forall j \in [d], t^*_i[A_j] = t_i[A_j] + \Lap\left(\frac{2d}{\epsilon}\right),$$
where $\Lap(\lambda)$ denotes a random variable that follows a Laplace distribution of scale $\lambda$, with the following probability density function:
$$pdf(x) = \frac{1}{2\lambda} \exp\left(-\frac{|x|}{\lambda}\right).$$

Once the aggregator receives all perturbed tuples, it simply computes their average $\frac{1}{n} \sum_{i = 1}^n t^*_i[A_j]$ as an estimate of the mean of $A_j$. Clearly, this estimate is unbiased, since the injected Laplace noise $\Lap\left(\frac{2d}{\epsilon}\right)$ in each $t^*_i[A_j]$ has zero mean. Meanwhile, it is easy to calculate that the expected error incurred by this estimator is $O\left(\frac{d}{\epsilon \sqrt{n}}\right)$, which is linear to the number of attributes $d$ and, thus, could be excessively large when there are many attributes. Note that this is a fundamental problem that also exists in the traditional differential privacy setting, when publishing statistics for multiple independent attributes. Perhaps rather surprisingly, this problem has not received much attention in the differential privacy literature; the first and only solution we are aware of is proposed by Duchi \emph{et al.}~\cite{DuchiJW14,Duchi2013} under the local differential privacy setting, presented below.

\header
{\bf Duchi \emph{et al.}'s method.}
 Algorithm~\ref{mech:pubduchi} shows the pseudo-code of Duchi \emph{et al.}'s method. The authors claim that this method satisfies $\epsilon$-local differential privacy, yields unbiased estimates for the mean value of each attribute, and incurs $O\big(\sqrt{d \log d} / (\epsilon \sqrt{n})\big)$ expected error for each attribute, which is proven to be asymptotically optimal. As we explain later, all three claims are incorrect, i.e., their method can violate differential privacy, lead to a biased estimate and incur a much higher amount of error.  It takes as input the exact tuple $t_i \in [-1, 1]^d$ of user $u_i$ and a privacy parameter $\epsilon$, and outputs a perturbed vector $t^*_i \in \{-B, B\}^d$, where $B$ is a constant decided by $d$ and $\epsilon$. Note that the output is binary for each attribute, i.e., it is either $B$ or $-B$. Therefore, it suffices for each user to transmit only one bit for each attribute to the aggregator.

Upon receiving the perturbed tuples, the aggregator simply computes the average value for each attribute over all users, and outputs these averages as the estimates of the mean values for their corresponding attributes. Next we focus on the calculation of $B$, which is rather complicated. Essentially, $B$ is a scaling factor to ensure that the expected value of a perturbed attribute is the same as that of the exact attribute value. First, we calculate:
\begin{equation} \label{eqn:basic-Cd}
C_d = \begin{cases}
2^{d-1}, &\text{if  $d$ is odd}\\
2^{d-1} - \frac{1}{2} \binom{d}{d/2}, &\text{otherwise}
\end{cases}
\end{equation}
Then, $B$ is calculated by:
\begin{equation} \label{eqn:basic-B}
B = \begin{cases}
\displaystyle \frac{2^d + C_d \cdot (e^\epsilon-1)}{  \binom{d-1} {(d-1)/2} \cdot (e^\epsilon - 1)}, &\text{if  $d$ is odd} \vspace{3mm} \\
\displaystyle \frac{2^d + C_d \cdot (e^\epsilon-1)}{ \binom{d-1} {d/2} \cdot (e^\epsilon - 1)}, &\text{otherwise}
\end{cases}
\end{equation}
Duchi \emph{et al.}~show that $\frac{1}{n} \sum_{i = 1}^n t^*_i[A_j]$ is an unbiased estimator of the mean of $A_j$, and
\begin{equation} \label{eqn:basic-error}
\E\left[\max_{j \in [d]} \left|\frac{1}{n} \sum_{i = 1}^n t^*_i[A_j] - \frac{1}{n} \sum_{i = 1}^n t_i[A_j]\right|\right] = O\left(\frac{\sqrt{d \log d}}{\epsilon \sqrt{n}}\right),
\end{equation}
which is asymptotically optimal \cite{DuchiJW14}.

\begin{algorithm}[t]
\begin{small}
\caption{Duchi \emph{et al.}'s Method \cite{DuchiJW14,Duchi2013}}\label{mech:pubduchi}
\SetKwInOut{Input}{input}
\SetKwInOut{Output}{output}
\Input{tuple $t_i \in [-1,1]^d$ and privacy parameter $\epsilon.$}
\Output{tuple $t^*_i \in \{-B, B\}^d.$}
    Generate a random tuple $v \in \{-1, 1\}^d$ by sampling each $v[A_j]$ independently from the following distribution:
    \parbox{62mm}{\begin{equation*}
    \Pr[v[A_j] = x] = \begin{cases}
        \frac{1}{2} + \frac{1}{2} t_i[A_j], & \text{if $x = 1$} \vspace{2mm} \\
        \frac{1}{2} - \frac{1}{2} t_i[A_j], & \text{if $x = -1$}\\
    \end{cases}
    \end{equation*}}\;
    Let $T^+$ (resp.\ $T^-$) be the set of all tuples $t^* \in \{-B, B\}^d$ such that $t^* \cdot v > 0$ (resp.\ $t^* \cdot v \le 0$)\;
    Sample a Bernoulli variable $u$ that equals $1$ with ${e^\epsilon / (e^\epsilon + 1)}$ probability\;
    \If{$u = 1$}
    {
        \Return a tuple uniformly at random from $T^+$\;
    }
    \Else
    {
        \Return a tuple uniformly at random from $T^-$\;
    }
\end{small}
\end{algorithm}

\header
{\bf Problems in Duchi \emph{et al.}'s method and a possible fix.} We implemented and evaluated Duchi \emph{et al.}'s method, but found that whenever the number $d$ of attribute is even, the method yields a biased estimation of the mean of each attribute and incurs significant error. Then, we also found that it violates differential privacy when $d$ is even. To illustrate, consider that $d=2$ and we have an input tuple $t_i = \langle 1, 1 \rangle$, i.e., $t_i[A_1] = t_i[A_2] = 1$. Then, Line 1 in Algorithm~\ref{mech:pubduchi} would generates a tuple $v = \langle 1, 1\rangle$. Let $B$ be as defined in Equation~\eqref{eqn:basic-B}, and $T^+$ and $T^-$ be as defined in Line 2 in Algorithm~\ref{mech:pubduchi}. It can be verified that $T^+$ and $T^-$ contain $1$ and $3$ tuples, respectively, with
\begin{align*}
T^+ &= \big\{\ \langle B, B \rangle \big\}, \textrm{ and} \\
T^- &= \big\{\ \langle -B, -B \rangle, \langle -B, B \rangle, \langle B, -B \rangle \big\}.
\end{align*}
Then, by Lines 3-8 in Algorithm~\ref{mech:pubduchi}, the method outputs $\langle B, B \rangle$ with $\frac{e^\epsilon}{e^\epsilon + 1}$ probability. In contrast, each tuple in $T^-$ has only $\frac{1}{3e^\epsilon + 3}$ probability to be output.

Now consider another input tuple $t^\prime_i = \langle -1, -1 \rangle$. It follows that, for $t^\prime_i$, the algorithm outputs $\langle B, B \rangle$ with only $\frac{1}{3e^\epsilon + 3}$ probability. As a consequence,
\begin{align*}
\Pr[f(t_i) = \big\langle B, B \rangle \big] & = 3 e^\epsilon \cdot \Pr\big[f(t^\prime_i) = \langle B, B \rangle\big] \\
& > e^\epsilon \cdot \Pr\big[f(t^\prime_i) = \langle B, B \rangle\big],
\end{align*}
which indicates that the algorithm does not satisfy $\epsilon$-differential privacy.

We find that the above problem is caused by Line 3 in Algorithm~\ref{mech:pubduchi}, in that the Bernoulli variable $u$ is incorrectly defined for the case for $d$ is even. To address the problem, one possible fix we found is to re-define $u$ as a Bernoulli variable such that
\begin{equation*}
\Pr[u = 1] = \frac{e^\epsilon \cdot C_d}{(e^\epsilon - 1)C_d + 2^d}.
\end{equation*}
It can be shown that, with this revised choice of $u$, Algorithm~\ref{mech:pubduchi} achieves $\epsilon$-differential privacy and ensures the error bound in Equation~\ref{eqn:basic-error}. We omit the proofs for brevity.

\header
{\bf Proposed method.} In what follows, we present an algorithm used in Harmony for perturbing a tuple that is conceptually simpler than Duchi \emph{et al.}'s method, but achieves the same privacy assurance and asymptotic error bound. Furthermore, our experiments (in Section~\ref{sec:exp}) show that the algorithm slightly outperforms Duchi \emph{et al.}'s method in terms of the empirical accuracy of the estimated means of numeric attributes. Additionally, our method is more efficient: in particular, each user only needs to transmit one bit to the aggregator, which is clearly optimal. Our algorithm is inspired by an existing approach \cite{BS15} for publishing categorical data, which we will discuss in Section~\ref{sec:basic-categorical}.

\begin{algorithm}[t]
\begin{small}
\caption{Proposed Method for Handling Numeric Attributes}\label{mech:pubreal}
\SetKwInOut{Input}{input}
\SetKwInOut{Output}{output}
\Input{tuple $t_i \in [-1,1]^d$ and privacy parameter $\epsilon.$}
\Output{tuple $t^*_i \in \left\{-\frac{e^\epsilon+1}{e^\epsilon-1}d, \:\: 0, \:\: \frac{e^\epsilon+1}{e^\epsilon-1} d\right\}^d.$}
    Let $t^*_i = \langle 0, 0, \ldots, 0\rangle$\;
    Sample $j$ uniformly at random from $[d]$\;
    Sample a Bernoulli variable $u$ such that
    \parbox{55mm}{$$\Pr[u = 1] = \frac{t_i[A_j]\cdot (e^\epsilon-1) + e^\epsilon + 1}{2e^\epsilon + 2}$$}\;
    \If{$u = 1$}
    {
        $t^*[A_j] = \frac{e^\epsilon+1}{e^\epsilon-1} \cdot d$\;
    }
    \Else
    {
        $t^*[A_j] = -\frac{e^\epsilon+1}{e^\epsilon-1} \cdot d$\;
    }
    \Return $t^*_i$
\end{small}
\end{algorithm}

Algorithm~\ref{mech:pubreal} shows the pseudo-code of our method. Given a tuple $t_i \in [-1, 1]^d$, the algorithm returns a perturbed tuple $t^*_i$ that has non-zero value on only one attribute $A_j$ ($j \in [d]$). Specifically, $A_j$ is selected uniformly at random from all $d$ attributes of $t_i$, and $t^*_i[A_j]$ is sampled from the following distribution:
\begin{equation} \label{eqn:basic-improved}
\Pr\big[t^*_i[A_j] = x\big] =
\begin{cases}
\frac{t_i[A_j]\cdot (e^\epsilon-1) + e^\epsilon + 1}{2e^\epsilon + 2}, & \textrm{if $x = \frac{e^\epsilon+1}{e^\epsilon-1} \cdot d$} \vspace{2mm} \\

\frac{-t_i[A_j]\cdot (e^\epsilon-1) + e^\epsilon + 1}{2e^\epsilon + 2}, & \textrm{if $x = -\frac{e^\epsilon+1}{e^\epsilon-1} \cdot d$}
\end{cases}
\end{equation}

Observe that in the above method, the output $t^*_i$ contains only one non-zero value, for the randomly chosen attribute $A_j$. This value is binary; hence, the user $u_i$ only needs to transmit 1 bit to the aggregator indicating its sign, and the latter can re-scale it using parameters $\epsilon$ and $d$. Further, as we show below the correctness of this method does not depend on the choice of $A_j$ as long as it is chosen uniformly at random. Therefore, the value of $j$ can be obtained, e.g., using a public source of random numbers such as a hash value of the user's ID. Therefore, the communication overhead between each user and the aggregator is exactly 1 bit.

The following lemmas establish the theoretical guarantees of Algorithm~\ref{mech:pubreal}.
\begin{lemma} \label{lmm:basic-privacy}
Algorithm~\ref{mech:pubreal} satisfies $\epsilon$-local differential privacy.
\end{lemma}
\begin{proof}
Let $t^*$ be an output of Algorithm~\ref{mech:pubreal}, and $A_j$ be the only attribute such that $t^*[A_j] \ne 0$. Let $t$ and $t^\prime$ be any two tuples, and $u$ (resp.\ $u^\prime$) be the Bernoulli variable generated in Line 3 of Algorithm~\ref{mech:pubreal} given $t$ (resp.\ $t^\p)$ as the input. In the following, we focus on the case when $t^*[A_j] = \frac{e^\epsilon + 1}{e^\epsilon - 1}d$; the case when
$t^*[A_j] = -\frac{e^\epsilon + 1}{e^\epsilon - 1}d$ can be analyzed in a similar manner.

By Algorithm~\ref{mech:pubreal}, we have
\begin{align*}
\frac{\Pr [t^* \mid t]}{\Pr[t^* \mid t^\p]} & = \frac {1/d \cdot \Pr[u = 1 \mid t]}{1/d \cdot \Pr[u^\p = 1 \mid t^\p] } \; \leq \frac {\max_t \Pr[u = 1 \mid t]} {\min_{t^\p} \Pr[u^\p = 1 \mid t^\p]}\\
& = \frac{\max_{t[A_j]\in [-1,1]} \left(t[A_j] \cdot (e^\epsilon-1) + e^\epsilon + 1\right)}{\min_{t^\p[A_j]\in [-1,1]} \left(t^\p[A_j]\cdot (e^\epsilon-1) + e^\epsilon + 1\right)}  = e^\epsilon.
\end{align*}
This completes the proof.
\end{proof}

\begin{lemma} \label{lmm:basic-unbias}
Let $t^*_i$ be the output of Algorithm~\ref{mech:pubreal} given an input tuple $t_i$. Then, for any $j \in [d]$, $\E[t^*[A_j]] = t[A_j]$.
\end{lemma}
\begin{proof}
By Equation~\eqref{eqn:basic-improved},
\begin{align*}
\E[t^*[A_j]] = & \textstyle \Pr\left[t^*[A_j] = \frac{e^\epsilon+1}{e^\epsilon-1}\cdot d\right] \cdot \frac{e^\epsilon+1}{e^\epsilon-1}  \cdot d\\
  & {} \textstyle + \Pr\left[t^*[A_j] = -\frac{e^\epsilon+1}{e^\epsilon-1} \cdot d\right] \cdot \left(- \frac{e^\epsilon+1}{e^\epsilon-1}  \cdot d \right)  \\
  & {} \textstyle + \Pr\left[t^*[A_j] = 0\right] \cdot 0 \\
= & \frac{2 t[A_j] \cdot (e^\epsilon - 1)}{2e^\epsilon - 2} \; = t[A_j].
\end{align*}
\end{proof}

By Lemma~\ref{lmm:basic-unbias}, the server can use $\frac{1}{n} \sum_{i = 1}^n t^*[A_j]$ as an unbiased estimator of the mean of $A_j$. The following lemma shows the accuracy guarantee of this estimator.
\begin{lemma} \label{lmm:basic-accuracy}
For any $j \in [d]$, let $Z[A_j] = \frac{1}{n} \sum_{i=1}^n t^*_i[A_j]$ and $X[A_j] = \frac{1}{n} \sum_{i=1}^n t_i[A_j]$. With at least $1 - \beta$ probability,
$$\max_{j \in [d]} \big|Z[A_j] - X[A_j]\big| = O\left(\frac{\sqrt{d \log(d/\beta)}}{\epsilon \sqrt{n}}\right).$$
\end{lemma}
\begin{proof}
First, observe that for any $i \in [d]$ and any $j \in [d]$, the variance of $t^*_i[a_j] - t_i[a_j]$ equals:
\begin{align*}
\Var[t^*_i[a_j] - t_i[a_j]] &= \Var[t^*_i[a_j]] \\
&= \E\left[(t^*_i[a_j])^2\right] - (\E[t^*_i[a_j])^2 \\
&= \textstyle \frac{1}{d} \left(\frac{e^\epsilon+1}{e^\epsilon-1} \cdot d\right)^2 - \left(t_i[a_j]\right)^2 \; \le  \textstyle \left(\frac{e^\epsilon+1}{e^\epsilon-1}\right)^2 \! \cdot d.
\end{align*}

By Bernstein's inequality,
\begin{align*}
& \Pr\big[|Z[a_j] - X[a_j]| \ge \lambda\big] \\
& {} \le 2 \cdot \exp\left(-\frac{n\lambda^2}{\frac 2 n \sum_{i=1}^n \Var[t^*_i[a_j] - t_i[a_j]] + \frac{2}{3}\lambda \cdot \frac{e^\epsilon + 1}{e^\epsilon-1} \cdot 2d}\right) \\
& {} = 2\cdot \exp \left(-\frac {n \lambda^2}{2d  \cdot \left(O(1/\epsilon^2) + \lambda \cdot O(1/\epsilon) \right)} \right).
\end{align*}
By the union bound, there exists $\lambda = O\left(\frac{\sqrt{d \log (d/\beta)}}{\epsilon \sqrt{n}}\right)$ such that $\max_{j \in [d]} |Z[a_j] - X[a_j]| < \lambda$
holds with at least $1-\beta$ probability.
\end{proof}

 \subsection{Estimating Frequencies for Categorical Attributes} \label{sec:basic-categorical}

We now focus on the case where each user's data record contains not only numeric attributes but also categorical ones. For each categorical attribute, the aggregator aims to build an accurate histogram containing the frequency estimate for each possible value in the attribute's domain. For example, Samsung may want to know the percentage of users who enable a specific setting, through the diagnostic information report app described in Section \ref{sec:usecase}. Note that we can convert a numeric attribute a categorical one (e.g., display brightness can be discretized to three levels: low, medium and high) and build a histogram accordingly.

\header
\textbf{Randomized response for binary attributes.} For a single binary attribute (e.g., WiFi on/off), it suffices to use the classic randomized response method \cite{Warner65} (also used in Rappor \cite{RAPPOR2014}) to estimate the distribution of users. Specifically, suppose that the domain of the binary attribute (let $A_j$) contains two possible values, $-1$ and $+1$. Each user $u_i$ reports her true answer $t_i[A_j]$ with probability $p$, and a random answer with probability $1-p$. The latter has the same probability to be $-1$ and $+1$; hence, its expected value is zero. Therefore, the expected value for $u_i$'s reported value is $p \cdot t_i[A_j]$; thus, we can obtain an unbiased estimate by multiplying the reported value by a scaling factor $c_\epsilon = 1/p$.

Meanwhile, comparing $u_i$'s true attribute value and her reported one, the two are the same with probability ${p+(1-p)/2}$, and they are different with probability $(1-p)/2$. According to Definition \ref{def:ldp}, $\epsilon$-local differential privacy requires that $\frac{p+(1-p)/2}{(1-p)/2} \leq e^{\epsilon}$. The equality holds when $p = \frac{e^\epsilon - 1}{e^\epsilon + 1}$. We thus arrive at the following unbiased mechanism that satisfies $\epsilon$-LDP: each user $u_i$ reports $c_\epsilon \cdot t_i[A_j] = \frac{e^\epsilon + 1}{e^\epsilon - 1}\cdot t_i[A_j]$ with probability $p+(1-p)/2 = \frac{e^\epsilon}{e^\epsilon + 1}$, and $-c_\epsilon \cdot t_i[A_j]$ otherwise (i.e., with probability $\frac{1}{e^\epsilon + 1}$).

Once the aggregator receives all reported values for attribute $A_j$, it computes the average over all users, which is an estimate of the mean value $\E[A_j]$ for $A_j$. Since $A_j$ can be either $+1$ or $-1$, the percentage of users with $+1$ (resp. $-1$) is $\frac{1+\E[A_j]}{2}$ (resp. $\frac{1-\E[A_j]}{2}$).

\begin{algorithm}[t]
\begin{small}
\caption{Bassily and Smith's method \cite{BS15}}\label{mech:pubbinvec}
\SetKwInOut{Input}{input}
\SetKwInOut{Output}{output}
\Input{$t_i[A_j] \in [k]$ for each user $u_i$, privacy budget $\epsilon$, confidence of the error bound $\beta$}
\Output{Frequency estimate for each of the $k$ values in attribute $A_j$}
    Compute $\gamma = \sqrt{\frac{\log(2k/\beta)}{\epsilon^2 n}}$\, and $m = \frac{\log(k+1) \log(2/\beta)}{\gamma^2}$\;
    Generate random matrix $\Phi \in \{\pm \frac 1 {\sqrt m}\}^{m \times k}$\;
    \For{$i$ = $1$ to $n$}
    {
        User $u_i$: draw $s ~\sim \text{Uniform}(\{1,2\dots, m\})$\;
        User $u_i$: draw $t \sim \text{Bern}(\frac{e^\epsilon}{e^\epsilon + 1})$\;
        User $u_i$: \textbf{if} $t=1$ \textbf{then} $\alpha = c_\epsilon m \Phi[s, t_i[A_j]]$ \textbf{else} $\alpha = -c_\epsilon m \Phi[s, t_i[A_j]]$, where $c_\epsilon = \frac{e^\epsilon+1}{e^\epsilon - 1}$\;
        User $u_i$: submit $\langle s, \alpha\rangle$, which represents a $k$-dimensional vector $z_i$ where the $s$-th entry is $\alpha$ and the other entries are $0$\;
    }
    Compute $\bar z = \frac 1 n \sum_{i=1}^n z_i$\;
    \For{$l$ = $1$ to $k$}
    {
       Estimate the frequency of the $l$-th value by the inner product of the $l$-th column of $\Phi$ and $\bar z$;
    }
    \Return $k$ frequency estimates obtained above;
\end{small}
\end{algorithm}

\header
\textbf{Bassily and Smith's method.} The problem is more complicated when the categorical attribute $A_j$ contains $k>2$ possible values. In this situation, the aggregator aims to build a histogram that contains the estimated frequency for each of the $k$ possible values. The current state-of-the-art to our knowledge is by Bassily and Smith \cite{BS15}, shown in Algorithm \ref{mech:pubbinvec}, which is proven to satisfy $\epsilon$-LDP and achieve an optimal asymptotical error bound. There are two main ideas in this method. First, the authors assume that the number of possible values $k$ in the categorical attribute is far larger than the number of users $n$; hence, the method applies random projection to reduce the dimensionality from $k$ to $m$. The value of $m$ is chosen carefully so as to obtain the asymptotically optimal error bound. This step essentially transforms the categorical attribute into $m$ binary ones.

Specifically, the random projection is done with a $m \times k$ matrix $\Phi$ in which each element is randomly set to either $+\frac{1}{\sqrt m}$ or $-\frac{1}{\sqrt m}$ with equal probability. This ensures that (i) the inner product of any column in $\Phi$ with itself is $1$ (which is where the absolute value of each element $\frac{1}{\sqrt m}$ comes from) and (ii) the inner product of two different columns in $\Phi$ has zero expected value, since the signs are randomly generated. Each user $u_i$'s attribute value $t_i[A_j]$ is then transformed to $m$ binary values by taking the $t_i[A_j]$-th column in $\Phi$.

The second idea is for each user to randomly pick one of the converted $m$ binary attributes, and report a randomized response using the method described earlier. Note that the randomized response needs to be scaled by a factor of $m$, since each of the $m$ binary attributes has probability $1/m$ to be chosen. The aggregator collects the average for each of the $m$ binary attributes, which is stored as a vector $\bar z$. To obtain the frequency estimate of a particular attribute value $l$, the method takes the inner product of $\bar z$ and the $l$-th column of $\Phi$, which can be proven to yield an unbiased frequency estimate for the $l$-th value in $A_j$. Meanwhile, the frequency estimate is proven to be within $O\left(\frac{\sqrt{\log(k/\beta)}}{\epsilon \sqrt{n}}\right)$ error with probability $1-\beta$ \cite{BS15}, where $\beta$ is an input to the algorithm.

\header
\textbf{Proposed method for a single categorical attribute.} Harmony generally follows Bassily and Smith's method to estimate value frequencies for a categorical attribute. However, we found that although Bassily and Smith's method achieves optimal asymptotic accuracy, in practice its accuracy tends to be unstable, especially for relatively small categorical domains. The reason is that the random projection matrix $\Phi$ introduces considerable noise; in particular, the inner product of two different columns is often non-zero unless $k$ is very large, which is magnified by a large number of users $n$. Hence, we propose an alternative solution that obtains higher accuracy when $k = o(n)$. In particular, instead of generating random matrix $\Phi$ (size $m \times k$ where $m = O(n)$), we construct a binary matrix of size $k \times k$ satisfying that any two column vectors are always orthogonal. The construction algorithm of this matrix can be found in the appendix.

\header
\textbf{Proposed method for multiple numeric and categorical attributes.}
Bassily and Smith's method is limited to a single categorical attribute. To extend it to multiple categorical attributes, a straightforward approach is to apply the method once for each attribute separately. In that case, however, the privacy budget $\epsilon$ needs to be divided among all attributes, so as to ensure $\epsilon$-LDP as a whole. Without loss of generality, assume that we have $d$ categorical attributes, and we assign $\epsilon/d$ budget to each of them. Then, the amount of noise incurred by Bassily and Smith's method on each attribute is increased $d$ times to $O\left(\frac{d\sqrt{\log(d/\beta)}}{\epsilon \sqrt{n}}\right)$, which is unsatisfactory when $d$ is large.

Another approach for extension is to (i) convert the $d$ categorical attributes into a ``composite'' attribute whose domain equals the Cartesian product of the individual attribute domains, and then (ii) apply Bassily and Smith's method on the composite attribute. This, however, only allows us to (accurately) derive the frequency each composite value (i.e., combination of values from all $d$ individual attributes), but does not provide quality estimation of the frequency of each individual value. Note that these limitations are not specific to Bassily and Smith's method; to our knowledge, there is no existing work (including Rappor) that are designed to handle multiple categorical attributes, let alone a mixture of numeric and categorical ones.

In Harmony, we use a simple and elegant solution to handle multiple attributes: for each numerical attribute, the aggregator estimates its mean value; for each categorical attribute, the aggregator estimates its value frequencies. In particular, given $d$ attributes $A_1, A_2, \ldots, A_d$, the solution asks each user to perform the following:
\begin{enumerate}[topsep = 6pt, parsep = 6pt, itemsep = 0pt, leftmargin=22pt]
    \item Draw $j$ uniformly at random from set $\{1, 2, \ldots, d\}$;
    \item If $A_j$ is a numeric attribute, then submit a noisy version of $t[A_j]$ computed as in Lines $3$-$7$ of Algorithm~\ref{mech:pubreal};
    \item Otherwise (i.e., $A_j$ is a categorical attribute), compute ${\langle s, \alpha\rangle}$ as in Lines $4$-$8$ of Algorithm~\ref{mech:pubbinvec}, then submits ${\langle s, d \cdot \alpha\rangle}$ to represent a $d$-dimensional vector where the $s$-th entry is $d \cdot \alpha$ and all other entries are zero.
\end{enumerate}

The above solution satisfies $\epsilon$-LDP, which follows from the fact that (i) each user randomly selects one attribute to submit, and (ii) the algorithm used for submitting the selected attribute is $\epsilon$-differentially private. For each numeric attribute, it is easy to see that the solution provides the same accuracy guarantee as Algorithm~\ref{mech:pubreal}, since both methods handle numeric attributes in exactly the same way. The following lemma states the accuracy guarantee of our solution for categorical attributes.

\begin{lemma} \label{lmm:basic-accuracy2}
For each categorical attribute $A_j$ with a domain $\{1, 2, \ldots, k\}$, let $x_l$ be the frequency of the $l$-th value of $A_j$, and $y_l$ be the estimation of $x_l$ returned by our solution. With at least $1 - \beta$ probability,
$$\max_{l \in [k]} \big|y_l - x_l\big| = O\left(\frac{\sqrt{d \log(k/\beta)}}{\epsilon \sqrt{n}}\right).$$
\end{lemma}
\begin{proof}[(Sketch)] Consider any user with a tuple $t$. With respect to $A_j$, our solution can be regarded as a method that (i) outputs nothing with ${d-1}/d$ probability, and (ii) with the remaining $1/d$ probability, applies Bassily and Smith's algorithm on $A_j$ and scale its output up by $d$ times. It follows that the variance of each of our frequency estimators for $A_j$ is $d$ times that of Bassily and Smith's estimator. Based on the analysis in \cite{BS15}, it can be shown that the variance of Bassily and Smith's estimator is $O\left(\frac{1}{n \epsilon^2}\right)$. Therefore, the variance of each of our frequency estimators is $O\left(\frac{d}{n \epsilon^2}\right)$. Combining this with Bernstein's inequality and the union bound, it can be proven that with at least $1-\beta$ probability, the error in any $k$ of our estimators is $O\left(\frac{\sqrt{d \log(k/\beta)}}{\epsilon \sqrt{n}}\right).$
\end{proof}

By Lemma~\ref{lmm:basic-accuracy2}, when there exist multiple categorical attributes, the error incurred by our approach is a factor of $O(\sqrt{d})$ smaller than that of a solution that repeatedly apply Bassily and Smith's method on each categorical attribute.

 \section{Building Machine Learning Models using Stochastic Gradient Descent} \label{sec:riskmini}

This section investigates building a large class of machine learning models that can be expressed as empirical risk minimization under $\epsilon$-local differential privacy. In particular, we focus on three common types of learning tasks: linear regression, logistic regression, and SVM classification. Section~\ref{sec:riskmini-methods} introduces the basic approaches for building these models, while Section~\ref{sec:riskmini-improve} discusses optimizations that lead to improved results in practice.

\subsection{Basic Methods} \label{sec:riskmini-methods}

Suppose that each user $u_i$ has a pair $\langle x_i, y_i\rangle$, where $x_i \in [-1, 1]^d$ and $y_i \in [-1, 1]$ (for linear regression) or $y_i \in \{-1, 1\}$ (for logistic regression and SVM classification). Let $\ell(\cdot)$ be a {\em loss function} that (i) maps a $d$-dimensional {\em parameter vector} $\beta$ into a real number and (ii) is parameterized by $x_i$ and $y_i$. We aim to identify a parameter vector $\beta^*$ such that

\vspace{-2mm}
$$\beta^* = \arg\min_{\beta} \frac{1}{n}\left(\sum_{i=1}^n \ell(\beta; x_i, y_i)\right) + \frac{\lambda}{2} \|\beta\|^2_2,$$
where $\lambda > 0$ is a regularization parameter. We consider three specific loss functions:
\begin{enumerate}[topsep = 6pt, parsep = 6pt, itemsep = 0pt, leftmargin=22pt]
\item Linear regression: $\ell(\beta; x_i, y_i) = (x_i^T \beta - y_i)^2$;

\item Logistic regression: $\ell(\beta; x_i, y_i) = \log\left(1 + e^{-y_ix_i^T\beta}\right)$;

\item SVM (hinge loss): $\ell(\beta; x_i, y_i) = \max\{0, 1 -y_i x_i^T\beta\}$.
\end{enumerate}
For convenience, we define
$$\ell^\prime(\beta; x_i, y_i) = \ell(\beta; x_i, y_i) + \frac{\lambda}{2} \|\beta\|^2_2.$$

One of the most common solutions to compute $\beta^*$ is {\em stochastic gradient descent (SGD)}. It starts from an initial parameter vector $\beta_0$, and iteratively updates it into $\beta_1, \beta_2, \ldots$ based on the following equation:
$$\beta_{k+1} = \beta_k - \gamma_k \cdot \nabla \ell^\prime(\beta_k; x, y),$$
where $\langle x, y \rangle$ is the tuple of a randomly selected user, $\nabla \ell^\prime(\beta_k; x, y)$ is the gradient of $\ell^\prime$ at $\beta_k$, and $\gamma_k$ is a constant typically set to $O(1/\sqrt{k})$. It terminates when the difference between $\beta_{k+1}$ and $\beta_k$ is sufficiently small.

Under our problem setting, however, $\nabla \ell^\prime$ is not directly available to the aggregator, and needs to be collected in a private manner. Towards this end, existing work \cite{HammCCBX15,DuchiJW14} has suggested that the aggregator may ask the selected user in each iteration to submit a noisy version of $\nabla \ell^\prime$, by using the Laplace mechanism or Duchi \emph{et al.}'s method (i.e., Algorithm~\ref{mech:pubduchi}). We can straightforwardly improve these existing approaches by perturbing $\nabla \ell^\prime$ using Algorithm~\ref{mech:pubreal} instead; however, we observe that such a solution is insufficient for our target application, as we explain in Section~\ref{sec:riskmini-improve}.

\subsection{Improvements} \label{sec:riskmini-improve}

\header{\bf \noindent Mini-batching.} We observe in our experiments that the aforementioned SGD approach yields rather inaccurate results, due to the noise injected in the gradient $\nabla \ell^\prime$ returned by each user. In particular, if each user applies Algorithm~\ref{mech:pubreal} to compute $\nabla \ell^\prime(\beta_k; x, y)$, the amount of noise in $\nabla \ell^\prime$ is $O\left(\frac{\sqrt{d \log d}}{\epsilon}\right)$, which is excessively large given that $x \in [-1, 1]^d$. To address this issue, we adopt {\em mini-batch} gradient descent instead of SGD. That is, each iteration of the algorithm, we involve a group $G$ of users, and ask each of them to submit a noisy version of the gradient; after that, we update the parameter vector $\beta_k$ with the mean of the noisy gradients, i.e.,
$$\textstyle \beta_{k+1} = \beta_k - \gamma_k \cdot \frac{1}{|G|}\sum_{i=1}^{|G|} \nabla \ell^*_i,$$
where $\nabla \ell^*_i$ is the noisy gradient submitted by the $i$-th user in $G$. This helps because the amount of noise in the average gradient is $O\left(\frac{\sqrt{d \log d}}{\epsilon \sqrt{|G|}}\right)$, which could be acceptable if $|G| = \Omega\left(d \log d/\epsilon^2\right)$.

However, when $d$ is sizable, $|G| = \Omega\left(d \log d/\epsilon^2\right)$ is large. As a consequence, when we allow each user to participate in at most one iteration of the algorithm, the maximum number of iterations (i.e., $n/|G|$) is small. In that case, the algorithm may terminate prematurely and return an inferior parameter vector. One may attempt to mitigate this problem by allowing each user to be involved in $m > 1$ iterations, but it would further increase the amount of noise in each gradient returned. To explain this, suppose that the $i$-th ($i \in [1, m]$) gradient returned by the user satisfies $\epsilon_i$-differential privacy. By the composition property of differential privacy \cite{McSherryT07}, if we are to enforce $\epsilon$-differential privacy for the user's data, we should have $\sum_{i=1}^m \epsilon_i \le \epsilon$. Consider that we set $\epsilon_i = \epsilon/m$. Then, the amount of noise in each gradient becomes $O\left(\frac{m \sqrt{d \log d}}{\epsilon}\right)$; accordingly, the acceptable mini-batch size becomes $|G| = \Omega\left(m^2 d \log d/\epsilon^2\right)$, which is $m^2$ times the acceptable size when we allow each user to participate in at most one iteration. It then follows that the total number of iterations in the algorithm is inversely proportional to $1/m$, i.e., setting $m > 1$ only degrades the performance of the algorithm.

\header{\bf \noindent Dimension reduction.} For linear regression, instead of increasing $m$, we propose to apply dimensionality reduction on each user's data, so as to reach an acceptable size of mini-batches. Specifically, the curator first generates a random $r \times d$ matrix $P$ where $r < d$ and each entry has an equal probability to be assigned $1/d$ or $-1/d$. Then, the curator shares $P$ with all users, and asks each user to convert her tuple $\langle x_i, y_i\rangle$ into a reduced tuple $\langle x^\p_i, y_i\rangle$, where $x^\p_i = P x$. In other words, we project $\{x_i\}$ into a random $r$-dimensional sub-space, and such a projection is known to preserve several important characteristics of the original data \cite{Achlioptas01}. It can be verified that $x^\p_i \in [-1, 1]^r$.

Subsequently, each user uses the reduced tuple $\langle x^\p_i, y_i\rangle$ to participate in the mini-batch gradient descent algorithm. In other words, each noisy gradient $\nabla \ell^*_i$ returned by the user is $r$-dimensional instead of $d$-dimensional. As such, the average noisy gradient obtained from a mini-batch $G$ of users has an error of $O\left(\frac{\sqrt{r \log r}}{\epsilon \sqrt{|G|}}\right)$ instead of $O\left(\frac{\sqrt{d \log d}}{\epsilon \sqrt{|G|}}\right)$. Accordingly, the acceptable mini-batch size is reduced to $|G| = \Omega\left(r \log r/\epsilon^2\right)$.

Algorithm~\ref{mech:pubiter} shows the pseudo-code of our mini-batch gradient descent method with dimension reduction, in the context of the Samsung diagnostic tool. The aggregator first generates a random $r \times d$ matrix $P$, and maintains a $r$-dimensional parameter vector $\alpha$ (Lines 1-3). (We use $\alpha$ instead of $\beta$ to denote the parameter vector to avoid confusion on its dimensionality.) After that, whenever a user with a tuple $\langle x_i, y_i\rangle$ comes online, she obtains $P$ and the current $\alpha$ from the aggregator (Line 6). Then, the user computes a reduced tuple $\langle x^\p_i, y_i\rangle$, as well as the gradient $\nabla_i = \nabla \ell^\p(\alpha; x^\p_i, y_j)$ (Line 7). If any entry of $\nabla_i$ is larger than $1$ (resp.\ smaller than $-1$), then the user resets the entry to $1$ (resp.\ $-1$) (Lines 8-9). This ensures that $\nabla_i \in [-1, 1]^d$, so that it can be a valid input to Algorithm~\ref{mech:pubreal}. After that, the user computes a noisy gradient $\nabla^*_i$ using Algorithm~\ref{mech:pubreal}, submits it to the aggregator, and then logs off (Line 10).

The aggregator computes the average noisy gradient from every $g$ users (where $g$ is an input parameter), and updates the parameter vector $\alpha$ accordingly (Lines 11-14). When the update to $\alpha$ is sufficiently small or when a sufficiently large number of users have participated, the aggregator terminates the algorithm (Lines 15-16).

\begin{algorithm}[t]
\begin{small}
\caption{An Iterative Method for Empirical Risk Minimization}\label{mech:pubiter}
\SetKwInOut{Input}{input}
\SetKwInOut{Output}{output}
\Input{positive integer $r$, mini-batch size $g$, privacy parameter $\epsilon$}
\Output{parameter vector $\alpha \in \mathbb{R}^r$}
    generates a random $r \times d$ matrix $P$ each entry has an equal probability to $1/d$ or $-1/d$\;
    initialize a counter $k = 0$, and learning rate $\gamma$\;
    initialize a $r$-dimensional vector $\nabla = \langle 0, 0, \ldots, 0\rangle$\;
    \While{true}
    {
        \If{a user with a tuple $\langle x_i, y_i\rangle$ comes online}
        {
            send $P$ to the user\;
            the user computes $x^\p_i = P x_i$, and derives $\nabla_i = \nabla \ell^\p(\alpha; x^\p_i, y_j)$\;
            \If{$\nabla_i \notin [-1, 1]^r$}
            {
                the user projects $\nabla_i$ onto $[-1, 1]^r$\;
            }
            the user applies Algorithm~\ref{mech:pubreal} on $\nabla_i$, and submits a noisy gradient $\nabla^*_i$\;
            $k = k + 1$, and $\nabla = \nabla + \nabla^*_i$\;
            \If{$k \mod g = 0$}
            {
                $\nabla = \frac{\nabla}{g}$, and $\gamma = 1/\sqrt{k/g}$\;
                $\alpha = \alpha - \gamma \cdot \nabla$\;
            }
            \If{$k$ is sufficiently large or $\alpha$ changes sufficiently small in the last update}
            {
                {\bf break}\;
            }
        }

    }
    \Return $\alpha$\;
\end{small}
\end{algorithm}

\begin{figure*}[t]
 \large
    \centering
\begin{subfigure}[b]{0.24\textwidth}
\begin{tikzpicture}[scale=0.5]
\begin{axis}[xmode=log,ymin=0,ymax=1.5,xtick={0.05, 0.1, 0.2, 0.4, 0.8},
xticklabels={0.05, 0.1, 0.2, 0.4, 0.8},
ylabel={$L_\infty$ error},xmin=0.05,
xmax=0.8,
xlabel={privacy budget $\epsilon$},]
\addplot [mark=o,mark size=4pt,color=black] coordinates{
(0.05, 0.204176 )
(0.1, 0.104995)
(0.2,  0.053417)
(0.4, 0.028152 )
(0.8, 0.014394 )};

\addplot [mark=square,mark size=4pt,color=black] coordinates{
(0.05,1.308343  )
(0.1, 0.620373)
(0.2,  0.310137)
(0.4, 0.147693 )
(0.8, 0.078879 )};

\addlegendentry{Our method}
\addlegendentry{Hybrid}

\end{axis}
\end{tikzpicture}
\caption{US - Categorical}
\end{subfigure}
\begin{subfigure}[b]{0.24\textwidth}
\begin{tikzpicture}[scale=0.5]
\begin{axis}[xmode=log,ymin=0,ymax=0.85,xtick={0.05, 0.1, 0.2, 0.4, 0.8},
xticklabels={0.05, 0.1, 0.2, 0.4, 0.8},
ylabel={$L_\infty$ error},xmin=0.05,
xmax=0.8,
xlabel={privacy budget $\epsilon$},]
\addplot [mark=o,mark size=4pt,color=black] coordinates{
(0.05, 0.204090 )
(0.1, 0.096314)
(0.2,  0.046303)
(0.4, 0.023612 )
(0.8, 0.013039 )};

\addplot [mark=square,mark size=4pt,color=black] coordinates{
(0.05,0.803837  )
(0.1, 0.355229)
(0.2,  0.206760)
(0.4, 0.092176 )
(0.8, 0.049845 )};

\addlegendentry{Our method}
\addlegendentry{Hybrid}

\end{axis}
\end{tikzpicture}
\caption{BR - Categorical}
\end{subfigure}
\begin{subfigure}[b]{0.24\textwidth}
\begin{tikzpicture}[scale=0.5]
\begin{axis}[xmode=log,ymin=0,ymax=0.07,xtick={0.05,0.1,0.2,0.4,0.8},
xticklabels={0.05,0.1,0.2,0.4,0.8},
ylabel={$L_\infty$ error},xmin=0.05,
xmax=0.8,
xlabel={privacy budget $\epsilon$},]

\addplot [mark=o,mark size=4pt,color=black] coordinates{
( 0.05  ,        0.0539382070707 )
( 0.1   ,        0.0266372106599 )
( 0.2   ,        0.0133722994315 )
( 0.4   ,        0.0068858221519 )
( 0.8   ,        0.00349457242252 )
};

\addplot [mark=square,mark size=4pt,color=black] coordinates{
( 0.05  ,        0.06685048625 )
( 0.1   ,        0.031690375625 )
( 0.2   ,        0.01639815375 )
( 0.4   ,        0.00790808125 )
( 0.8   ,        0.00414400875 )
};
\addlegendentry{Our method}
\addlegendentry{Hybrid}
\end{axis}
\end{tikzpicture}
\caption{US - Numeric}
\end{subfigure}
\begin{subfigure}[b]{0.24\textwidth}
\begin{tikzpicture}[scale=0.5]
\begin{axis}[xmode=log,ymin=0,xtick={0.05,0.1,0.2,0.4,0.8},
xticklabels={0.05,0.1,0.2,0.4,0.8},xmin=0.05,
xmax=0.8,
ylabel={$L_\infty$ error},
y tick label style={
    /pgf/number format/.cd,
        fixed,
        precision=2,
    /tikz/.cd
},
xlabel={privacy budget $\epsilon$},]

\addplot [mark=o,mark size=4pt,color=black] coordinates{
( 0.05  ,        0.0810361946734 )
( 0.1   ,        0.0405761457413 )
( 0.2   ,        0.0202238597716 )
( 0.4   ,        0.00992819381313 )
( 0.8   ,        0.00531713430013 )
};

\addplot [mark=square,mark size=4pt,color=black] coordinates{
( 0.05  ,        0.105065722222 )
( 0.1   ,        0.0521644610245 )
( 0.2   ,        0.0257926524664 )
( 0.4   ,        0.0124129617978 )
( 0.8   ,        0.00603039954853 )
};

\addlegendentry{Our method}
\addlegendentry{Hybrid}

\end{axis}
\end{tikzpicture}
\caption{BR - Numeric}
\end{subfigure}

    \caption{Experiments on estimating means and frequencies.}
    \label{fig:exp-mean-freq}
\end{figure*}
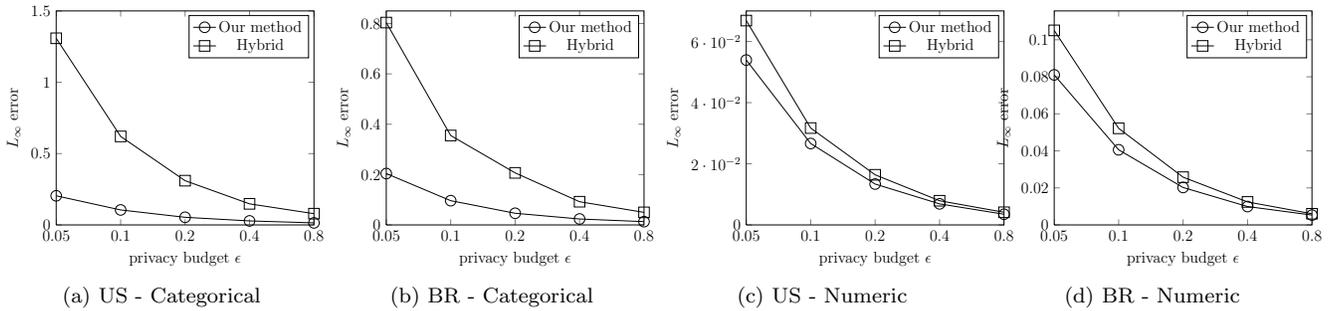

\section{Experiments} \label{sec:exp}

\subsection{Experimental Settings}
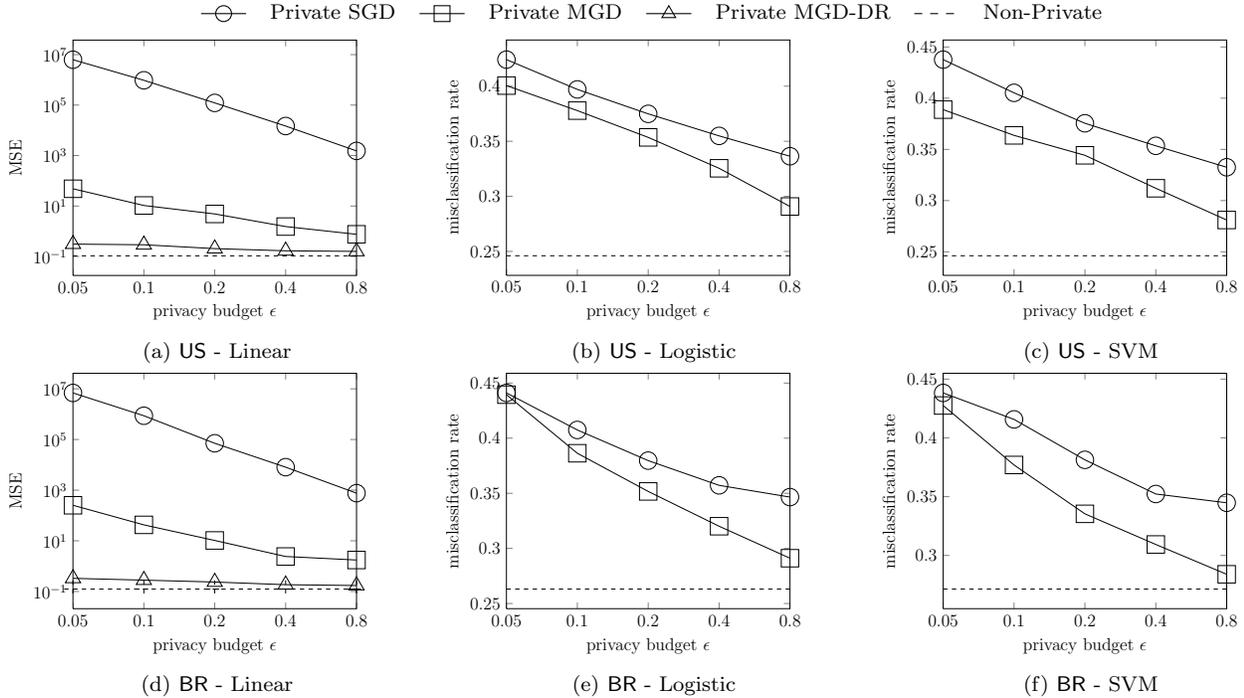
\begin{figure*}[t]
\large
\centering
 
   \begin{tikzpicture}
    \small
        \begin{customlegend}[legend columns=4,legend style={align=left,draw=none,column sep=2ex},legend entries={Private SGD, Private MGD, Private MGD-DR, Non-Private}]
        \addlegendimage{mark=o,mark size=3pt,solid,line legend}
        \addlegendimage{mark=square,mark size=3pt,solid}
        \addlegendimage{mark=triangle,mark size=3pt,solid}
        \addlegendimage{dashed,mark size=3pt}
        \end{customlegend}
     \end{tikzpicture}
     \hspace{10ex}
\begin{subfigure}[b]{0.32\textwidth}
\begin{tikzpicture}[scale=0.55]
\begin{axis}[
xtick={0.05,0.1,0.2,0.4,0.8},
xticklabels={0.05,0.1,0.2,0.4,0.8},
ymode = log,
xmin=0.05,
xmax=0.8,
xmode=log,
ylabel={MSE},
xlabel={privacy budget $\epsilon$},]

\addplot [mark=o, mark size=6pt,color=black] coordinates{
( 0.05     ,     6287122.48686 )
( 0.1     ,     960241.758476 )
( 0.2     ,     121961.847655 )
( 0.4     ,     14780.8136734 )
( 0.8     ,     1528.65032144 )
};

\addplot [mark=triangle, mark size=6pt,color=black] coordinates{
( 0.05  ,        0.3088838 )
( 0.1   ,        0.291185325 )
( 0.2   ,        0.2049374 )
( 0.4   ,        0.167753388889 )
( 0.8   ,        0.158758444444 )
};

\addplot [mark=square, mark size=6pt,color=black] coordinates{
( 0.05  ,        48.2389715575 )
( 0.1   ,        10.5001094074 )
( 0.2   ,        4.85417782407 )
( 0.4   ,        1.53365133333 )
( 0.8   ,        0.756012703704 )
};

\addplot [dashed, mark size=6pt,color=black] coordinates{
(0.05    ,      0.105894)
(0.1    ,     0.105894)
(0.2     ,    0.105894)
(0.4     ,    0.105894)
(0.8     ,    0.105894)
};

\end{axis}
\end{tikzpicture}
\caption{\textsf{US} - Linear}
\end{subfigure}
\begin{subfigure}[b]{0.32\textwidth}
\begin{tikzpicture}[scale=0.55]
\begin{axis}[
xtick={0.05,0.1,0.2,0.4,0.8},
xticklabels={0.05,0.1,0.2,0.4,0.8},
xmin=0.05,
xmax=0.8,
xmode=log,
ylabel={misclassification rate},
xlabel={privacy budget $\epsilon$},]

\addplot [mark=o, mark size=6pt,color=black] coordinates{
(0.05,0.424135041667)
(0.1,0.397075727273)
(0.2,0.374881722222)
(0.4,0.35487325)
(0.8,0.3364076875)
};

\addplot [mark=square, mark size=6pt,color=black] coordinates{
(0.05,0.400542926606)
(0.1,0.377820103774)
(0.2,0.35351575)
(0.4,0.32545315534)
(0.8,0.290787339806)
};

\addplot [dashed, mark size=6pt,color=black] coordinates{
(0.05,0.246063)
(0.1,0.246063)
(0.2,0.246063)
(0.4,0.246063)
(0.8,0.246063)};

\end{axis}
\end{tikzpicture}
\caption{\textsf{US} - Logistic}
\end{subfigure}
\begin{subfigure}[b]{0.32\textwidth}
\begin{tikzpicture}[scale=0.55]
\begin{axis}[
xtick={0.05,0.1,0.2,0.4,0.8},
xticklabels={0.05,0.1,0.2,0.4,0.8},
xmode=log,
xmin=0.05,
xmax=0.8,
ylabel={misclassification rate},
xlabel={privacy budget $\epsilon$},]

\addplot [mark=o, mark size=6pt,color=black] coordinates{
(0.05,0.4376538)
(0.1,0.405204425)
(0.2,0.3754819)
(0.4,0.353484657895)
(0.8,0.332550764706)
};

\addplot [mark=square, mark size=6pt,color=black] coordinates{
(0.05,0.388824658915)
(0.1,0.363690742188)
(0.2,0.344229)
(0.4,0.312033403226)
(0.8,0.281063243902)
};

\addplot [dashed, mark size=6pt,color=black] coordinates{
(0.05,0.246172)
(0.1,0.246172)
(0.2,0.246172)
(0.4,0.246172)
(0.8,0.246172)
};

\end{axis}
\end{tikzpicture}
\caption{\textsf{US} - SVM}
\end{subfigure}
\begin{subfigure}[b]{0.32\textwidth}
\begin{tikzpicture}[scale=0.55]
\begin{axis}[
xtick={0.05,0.1,0.2,0.4,0.8},
xticklabels={0.05,0.1,0.2,0.4,0.8},
ymode = log,
xmode=log,
xmin=0.05,
xmax=0.8,
ylabel={MSE},
xlabel={privacy budget $\epsilon$},]

\addplot [mark=o, mark size=6pt,color=black] coordinates{
( 0.05     ,     6889569.38021 )
( 0.1     ,     866543.579075 )
( 0.2     ,     70792.4341219 )
( 0.4     ,     8071.96353976 )
( 0.8     ,     760.210399583 )
};

\addplot [mark=triangle, mark size=6pt,color=black] coordinates{
( 0.05  ,        0.329723727941 )
( 0.1   ,        0.279884066667 )
( 0.2   ,        0.235770195489 )
( 0.4   ,        0.185013687023 )
( 0.8   ,        0.171521885496 )
};

\addplot [mark=square, mark size=6pt,color=black] coordinates{
( 0.05  ,        252.754903683 )
( 0.1   ,        42.5536999464 )
( 0.2   ,        10.3221948333 )
( 0.4   ,        2.40070314815 )
( 0.8   ,        1.74559125926 )
};

\addplot [dashed, mark=+, mark size=6pt,color=black] coordinates{
(0.05    ,      0.12416333333)
(0.1    ,     0.12416333333)
(0.2     ,    0.12416333333)
(0.4     ,    0.12416333333)
(0.8     ,    0.12416333333)
};

\end{axis}

\end{tikzpicture}
\caption{\textsf{BR} - Linear}
\end{subfigure}
\begin{subfigure}[b]{0.32\textwidth}
\begin{tikzpicture}[scale=0.55]
\begin{axis}[
xtick={0.05,0.1,0.2,0.4,0.8},
xticklabels={0.05,0.1,0.2,0.4,0.8},
xmin=0.05,
xmax=0.8,
xmode=log,
ylabel={misclassification rate},
xlabel={privacy budget $\epsilon$},]

\addplot [mark=o, mark size=6pt,color=black] coordinates{
(0.05,0.44105747)
(0.1,0.407447836735)
(0.2,0.379633322917)
(0.4,0.357206666667)
(0.8,0.346549819149)
};

\addplot [mark=square, mark size=6pt,color=black] coordinates{
(0.05,0.4394718)
(0.1,0.386306875)
(0.2,0.351691315)
(0.4,0.32011164)
(0.8,0.291095995)
};

\addplot [dashed, mark size=6pt,color=black] coordinates{
(0.05,0.263119833333333)
(0.1,0.263119833333333)
(0.2,0.263119833333333)
(0.4,0.263119833333333)
(0.8,0.263119833333333)
};

\end{axis}
\end{tikzpicture}
\caption{\textsf{BR} - Logistic}
\end{subfigure}
\begin{subfigure}[b]{0.32\textwidth}
\begin{tikzpicture}[scale=0.55]
\begin{axis}[
xtick={0.05,0.1,0.2,0.4,0.8},
xticklabels={0.05,0.1,0.2,0.4,0.8},
xmode=log,
xmin=0.05,
xmax=0.8,
ylabel={misclassification rate},
xlabel={privacy budget $\epsilon$},]

\addplot [mark=o, mark size=6pt,color=black] coordinates{
(0.05,0.43822852)
(0.1,0.41570793)
(0.2,0.381309575)
(0.4,0.3522503)
(0.8,0.344796745)
};

\addplot [mark=square, mark size=6pt,color=black] coordinates{
(0.05,0.4275431)
(0.1,0.376935815)
(0.2,0.335387385)
(0.4,0.309327135)
(0.8,0.283983395)
};

\addplot [dashed, mark size=6pt,color=black] coordinates{
(0.05,0.2714201875)
(0.1,0.2714201875)
(0.2,0.2714201875)
(0.4,0.2714201875)
(0.8,0.2714201875)
};

\end{axis}
\end{tikzpicture}
\caption{\textsf{BR} - SVM}
\end{subfigure}
\caption{Experimental results for empirical risk minimization.}\label{fig:exp-ERM}
\end{figure*}
 
For experimental repeatability, we use two public datasets extracted from the \textit{Integrated Public Use Microdata Series} \cite{IPUMS}, US and BR, which contains census records from the United States and Brazil, respectively. US contains $9$M tuples and $23$ attributes, among which $6$ are numeric (e.g., age) and $17$ are categorical (e.g., gender); BR has $4$M records and $18$ attributes, among which $6$ are numeric and $12$ are categorical. Both datasets contain a numeric attribute ``total income'', which we use as the dependent attribute in linear regression, logistic regression, and SVM. We normalize the domain of each numeric attribute into $[-1, 1]$.

\subsection{Estimating Means and Frequencies}

In the first set of experiments, we consider the task of collecting a noisy tuple from each user to estimate the mean of each numeric attribute and the frequency of each categorical value. As mentioned in Section~\ref{sec:intro}, none of the existing solutions can directly support this task, since they are designed for either numeric or categorical attributes, but not both. To address this issue, we construct a method (referred to as {\em Hybrid}) by combining the best existing solutions as follows. Let $t$ be a tuple with $d_n$ numeric attributes and $d_c$ categorical attributes. Given $t$ and a privacy budget $\epsilon$, {\em Hybrid} first constructs a $d_n$-dimensional tuple $t^\prime$ that contains all numeric values in $t$, and then release a noisy version of $t^\prime$ by invoking Duchi \emph{et al.}'s method (see Section~\ref{sec:basic-numeric}) on $t^\prime$, using a privacy budget of $d_n\epsilon/d$. After that, for each categorical value in $t$, {\em Hybrid} applies Bassily and Smith's method (see Section~\ref{sec:basic-categorical}) to release a noisy version with a privacy budget of $\epsilon/d$. By the composition property of differential privacy \cite{McSherryT07}, {\em Hybrid} ensures $\epsilon$-LDP. Intuitively, {\em Hybrid} is a best-effort approach to incorporate two states of the art that are designed only for numeric attributes (i.e., Duchi \emph{et al.}'s method) and a single categorical attribute (i.e., Bassily and Smith's method), respectively.

We apply our solution in Section~\ref{sec:basic} and {\em Hybrid} on both US and BR to generate noisy tuples, and then use the noisy tuples to estimate the frequency of each value in each categorical domain in US and BR. For each method, we measure the $L_\infty$ error in the estimated value frequencies, and we take the average measurement from $100$ runs. Figures \ref{fig:exp-mean-freq}a and \ref{fig:exp-mean-freq}b illustrate the results as $\epsilon$ varies. Observe that our method is significantly more accurate than {\em Hybrid} in all cases, and its error is only around $1/4$ of the error incurred by {\em Hybrid}. This is consistent with the analysis in Section~\ref{sec:basic-categorical} that (i) our method has $O\left(\frac{\sqrt{d\log k}}{\epsilon \sqrt{n}}\right)$ estimation error, and (ii) repeatedly applying Bassily and Smith's method on each categorical leads to $O\left(\frac{d\sqrt{\log k}}{\epsilon \sqrt{n}}\right)$ error.

We also use the noisy tuples to estimate the mean of each numeric attribute, and we measure the $L_\infty$ error of each method, averaged over $100$ runs. Figure \ref{fig:exp-mean-freq}c (resp.\ \ref{fig:exp-mean-freq}d) shows the results on US (resp.\ BR) as a function of the privacy budget $\epsilon$. Observe that our solution slightly outperforms {\em Hybrid}, regardless of the dataset used and the value of $\epsilon$. We also note that, compared with {\em Hybrid} (which applies Duchi \emph{et al.}'s method to handle numeric attributes), our solution has a much lower communication cost on each user (since it only requires each user to transfer $1$ bit), and is much simpler.

\subsection{Empirical Risk Minimization}

In the second set of experiments, we consider linear regression, logistic regression, and SVM classification on US and BR. For both datasets, we use the numeric attribute ``total income'' as the dependent variable, and all other attributes as independent variables. Following the standard practice, we transform each categorical attribute $A_j$ with $k$ values into $k-1$ binary attributes with a domain $\{-1, 1\}$, such that (i) the $l$-th ($l < k$) value in $A_j$ is represented by a $1$ on the
$l$-th binary attribute and a $-1$ on each of the remaining $k-2$ attributes, and (ii) the $k$-th value in $A_j$ is represented by $-1$ on all binary attributes. After this transformation, the dimensionality of US (resp.\ BR) becomes $85$ (resp.\ $95$). For logistic regression and SVM (which requires the dependent variable to be binary), we convert ``total income'' into a binary attribute by mapping all values which are greater than or equal to the mean to $1$, and all other values to $-1$.

We evaluate four methods: (i) a private version of SGD that involves one user in each iteration, and asks the user to submit a noisy gradient using Duchi \emph{et al.}'s method; (ii) mini-batch gradient descent (MGD), which involves $g$ users in each iteration, and uses the average noisy gradients of those users (generated with Algorithm~\ref{mech:pubreal}) to update the parameter vector; (iii) MGD with dimension reduction (DR), which is an improved version of MGD that projects the users' data onto an $r$-dimensional sub-space before the learning task (this method is applied for linear regression only); (iv) a non-private version of SGD. Based on the analysis in Section~\ref{sec:riskmini}, we set the mini-batch size for MGD and MGD-DR to $g = \max \{ 2d\log d /\epsilon^2, n/1000\}$ and $g = \max \{2r\log r / \epsilon^2, n/1000\}$, respectively. The term $n/1000$ is to guarantee that our mini-batch size is not too small when $\epsilon$ is large.  In addition, we set $r = 20$. For all methods, we set the regularization factor $\lambda = 10^{-4}$.

On each dataset, we use $10$-fold cross validation to assess the performance of each method. Figure \ref{fig:exp-ERM}a (resp.\ \ref{fig:exp-ERM}d) shows the mean squared error (MSE) of the linear regression model generated by each method, under various values of $\epsilon$. Private SGD incurs prohibitive errors in all cases, due to the large amount of noise in the gradient that it obtains in each iteration. MGD alleviates this issue with mini-batches, but still provides unsatisfactory accuracy for linear regression. The reason, as we mentioned in Section~\ref{sec:riskmini-improve}, is that MGD requires using a large mini-batch size $g$ when $d$ is large, which in turn leads to a small total number of iterations and degrades its performance. MGD-DR overcomes the drawback of MGD by incorporating dimension reduction to reduce the required mini-batch size, and hence, it is able to achieve an accuracy that is close to the non-private SGD.

Figures \ref{fig:exp-ERM}b and \ref{fig:exp-ERM}e (\ref{fig:exp-ERM}c and \ref{fig:exp-ERM}f) illustrate the misclassification rate of each method for logistic regression (resp.\ SVM). 
Overall, our experimental results demonstrate the effectiveness of mini-batches and dimension reduction for empirical risk minimization under local differential privacy.

 \section{Related Work} \label{sec:related}

Differential privacy \cite{dinur2003revealing, dwork2004privacy, DworkLaplace} is a strong, mathematically rigorous framework for privacy protection. Unlike earlier privacy-preserving data publication methods which are largely syntactic, differential privacy provides semantic, information-theoretic guarantees on individuals' privacy. Hence, since its proposal in 2003 it had attracted much attention from various research communities in computer science, including theory \cite{DworkR14}, machine learning \cite{SarwateC13}, data management \cite{YangZMWX12}, and systems \cite{ShiCRCS11}.

Earlier models of differential privacy assume a trusted data curator, who collects and manages the exact private information of individuals, and releases statistics derived from the data under differential privacy requirements. In practice, however, users may not want to share private information with anyone, including the central data curator. Recently, much attention has been shifted to the local differential privacy model, which eliminates the data curator and the collection of exact private information. Specifically, Duchi \emph{et al.} \cite{duchi2013local} systematically investigate the concept of local differential privacy, propose the minimax framework for LDP based on the information theory, prove upper and lower error bounds of LDP-compliant methods, and analyze the  trade-off between privacy and accuracy. Kairouz \emph{et al.}~\cite{kairouz2014differentially} show that a version of randomized response is an optimal mechanism for frequency estimation on a single binary attribute. Kairouz \emph{et  al.}~\cite{kairouz2014extremal} study the problem with a categorical attribute with an arbitrary number of possible values, propose two mechanisms: binary and randomized response mechanisms, and prove their optimality when the privacy budget is low and high, respectively. Bassily and Smith \cite{BS15} propose an asymptotically optimal solution for building succinct histograms over a large categorical domain under LDP.

Erlingsson \emph{et al.}~\cite{RAPPOR2014} propose the RAPPOR framework, which is based on the randomized response mechanism for publishing a vector of binary values under LDP. They use this mechanism with a Bloom filter, which intuitively adds another level of protection and increases the difficulty for the adversary to infer private information. As a result, it also becomes more difficult derive statistics from the collected data, and they propose a sophisticated solution for this purpose. A follow-up paper \cite{FantiPE15} extends Rappor to more complex statistics such as joint distributions and association testing, as well as categorical attributes that contain a large number of potential values, such as a user's home page.
 
\section{Conclusion}\label{sec:conclusion}

This work systematically investigates the problem of collecting and analyzing users' personal data under $\epsilon$-local differential privacy, in which the aggregator only collects randomized data from the users, and computes statistics based on such data. The proposed solution Harmony is able to collect data records that contain multiple numeric and categorical attributes, and compute accurate statistics from simple ones such as mean and frequency to complex machine learning models such as linear regression, logistic regression and SVM classification. Harmony achieves both optimal asymptotic error bound and high accuracy in practice. Meanwhile, it is highly efficient in terms of communication and computational overhead. Extensive experiments demonstrate its effectiveness on real data. In the next step, we plan to investigate the application of Harmony in a real use case such as Samsung's diagnostic info report app.

\bibliographystyle{abbrv}
\bibliography{main}

\balance

\appendix
\balance

\begin{algorithm}[t]
\begin{small}
\caption{Generation of orthogonal set}\label{mech:pubfixedduchi}
\SetKwInOut{Input}{input}
\SetKwInOut{Output}{output}
\Input{Dimensionality $d = 2^k.$}
\Output{Set $S$ of $d$ vectors.}
    $S =  \{ [1, -1], [1, 1] \}$\;
    \While {$|S| < d$}
    {
        $S^\prime = \emptyset$\;
        \For {$v \in S$}
        {
            $S^\prime \gets S^\prime \cup \{v|| v, v || (-v)\}$\;
        }
        $S \gets S^\prime$\;
    }
    \Return $S$\;
\end{small}
\end{algorithm}

\begin{lemma}The set $S$ returned from Algorithm~\ref{mech:pubfixedduchi} is an orthogonal set.
\end{lemma}
\begin{proof}
We prove this lemma by induction. For the base case, observe that the initial value of  $S =  \{ [1, -1], [1, 1] \}$ is a orthogonal set. Now assume that $S= \{v_i\} $ is a orthogonal set. We will prove that $S^\prime =  \{ v_i || v_i, v_i|| (-v_i)\}$ is also orthogonal set. For any $v^{\prime \prime}, v^\prime \in S^\prime$, consider the dot product $\langle v^{\prime \prime}, v^\prime \rangle$, there are two general cases: (i) $\langle v^{\prime \prime}, v^\prime \rangle = \langle v_i || \pm v_i, v_j || \pm v_j\rangle$ and (ii) $\langle v^{\prime \prime}, v^\prime \rangle = \langle v_i ||v_i ,v_i || - v_i\rangle$. In both cases, the inner product equals zero. Thus, lemma is proved.
\end{proof}

 \balance

\end{sloppy}

\end{document}